\documentclass[sn-mathphys]{sn-jnl}

\jyear{2021}%

\theoremstyle{thmstyleone}%
\newtheorem{theorem}{Theorem}
\usepackage[euler]{textgreek}
\usepackage{amscd,amssymb,amsmath,amsthm}
\usepackage{xcolor}
\usepackage{epstopdf}
\theoremstyle{thmstyletwo}%
\newtheorem{remark}{Remark}%

\theoremstyle{thmstylethree}%

\begin{document}

\title[The phase transition for the three-state SOS model...]{The phase transition for the three-state SOS model with one-level competing interactions on the binary tree}


\author[1,2,3]{\fnm{Muzaffar M.} \sur{Rahmatullaev}}\email{mrahmatullaev@rambler.ru}

\author*[4]{\fnm{Obid Sh.} \sur{Karshiboev}}\email{okarshiboevsher@mail.ru}
\equalcont{These authors contributed equally to this work.}

\affil[1]{\orgname{Institute of mathematics after named V.I.Romanovsky}, \orgaddress{\street{University street}, \postcode{100174}, \city{Tashkent}, \country{Uzbekistan}}}

\affil[2]{\orgname{New Uzbekistan University}, \orgaddress{\street{Mustaqillik Ave., 54}, \postcode{100007}, \city{Tashkent}, \country{Uzbekistan}}}

\affil[3]{\orgname{Namangan state university}, \orgaddress{\street{Uychi street, 316}, \postcode{160136}, \city{Namangan}, \country{Uzbekistan}}}

\affil*[4]{\orgname{Chirchik state pedagogical university}, \orgaddress{\street{Amir Temur street}, \city{Chirchik}, \postcode{111702}, \state{Tashkent region}, \country{Uzbekistan}}}


\abstract{In this paper, we consider a three-state solid-on-solid (SOS) model with two competing interactions (nearest-neighbour, one-level next-nearest-neighbour) on the Cayley tree of order two. We show that at some values of parameters the model exhibits a phase transition. We prove that for the model under some conditions there is no antiferromagnetic phases.}




\maketitle
\section{Introduction}

The solid-on-solid (SOS) model on a Cayley tree is introduced in \cite{rs} as a generalization of the Ising model. Since then a great interest has been devoted to the investigation of various properties of SOS models on Cayley trees (see, e.g., \cite{Sib,KRSOS,shok,Bun,Pos,PT,tmph}). See also \cite{Ro} and references therein for more details about SOS models on trees.

In this paper, we study the phase transition phenomenon for the three-state SOS model on a Cayley tree of order two with nearest-neighbour and one-level next-nearest-neighbour interactions. Note that the phase transition problem is one of the central problems of statistical mechanics \cite{G}. The existence of more than one Gibbs measure to a given model implies the occurrence of the phase transition \cite{G,Ro}. For the classical models (the Ising, Potts models) of statistical mechanics on Cayley trees within radius two interactions this problem is well studied (for the Ising model see, e.g., \cite{Van,Yokoi,Trag,Katsura}, for the Potts model see, e.g., \cite{GTA,Tur,Gan06,Pah2}).

We obtain a functional equation for the model using the self-similarity of the Cayley tree. Here we consider only the one-level next-nearest-neighbor interactions, since studying both one-level and prolonged next-nearest-neighbor interactions simultaneously usually lead to functional equations which difficult to solve (this happens even for the Ising model, see, e.g., \cite{ganipah}). We prove that at some values of parameters the model possess multiple Gibbs measures which implies the existence of phase transition. We show that for the model under certain conditions there is not any antiferromagnetic phase.  We also provide a conjecture on the absence of antiferromagnetic phase for the model on the invariant set.

 The paper is organized as follows. In Section \ref{defin} we give definitions of the model, Cayley tree and Gibbs measures. In Section \ref{recur} we reduce the problem of describing limit Gibbs measures to the problem of solving a system of nonlinear functional equations. Section \ref{tigm} is devoted to describe the ferromagnetic phase of the model. In Section \ref{per} we study the antiferromagnetic phase of the model.

\section{Preliminaries}\label{defin}

\textbf{Cayley tree.} The Cayley tree $\Gamma^k$ of order $k\geq1$ is an infinite tree, i.e., a cycles-free graph such that from each vertex of which issues exactly $k+1$ edges. We denote by $V$ the set of the vertices of tree and by $L$ the set of edges of tree. Two vertices $x$ and $y$, where $x,y\in V$ are called nearest-neighbor if there exists an edge $l\in L$ connecting them, which is denoted by $l=\langle x,y \rangle$. The distance on this tree, denoted by $d(x,y)$, is defined as the number of nearest-neighbour pairs of the minimal path between the vertices $x$ and $y$ (where path is a collection of nearest-neighbor pairs, two consecutive pairs sharing at least a given vertex).

For a fixed $x^0\in V,$ called the root, we set
$$
W_n=\{x\in V\mid d(x,x^0)=n\},~~~V_n=\bigcup_{m=0}^n W_m
$$
and denote by
$$
S(x)=\{y\in W_{n+1}:d(x,y)=1\},~~x\in W_n
$$
the set of \emph{direct successors} of $x.$
We observe that, for any vertex $x\neq x^0,$ $x$ has $k$ direct successors and $x^0$ has $k+1.$ For the sake of simplicity, we put $\mid x\mid=d(x, x^0),~x\in V$. Two vertices $x,y\in V$ are called second nearest-neighbor if $d(x,y)=2$. The second nearest-neighbor vertices $x$ and $y$ are called prolonged second nearest-neighbor if $\mid x\mid\neq\mid y\mid$ and is denoted by $>\widetilde{x,y}<.$ The second nearest-neighbor vertices $x,y\in V$ that are not prolonged are called one-level next-nearest-neighbor since $\mid x\mid=\mid y\mid$ and are denoted by $>\overline{x,y}<.$

In this paper, we consider a semi-infinite Cayley $\Gamma^k$ of order $k\geq2,$ i.e. a cycles-free graph with $(k+1)$ edges issuing from each except $x^0$ and with $k$ edges issuing from the vertex $x^0$. According to well known theorems, this can be reconstituted as a Cayley tree \cite{G,Tur}.

In the SOS model, the spin variables $\sigma(x)$ take their values on the set $\Phi=\{0,1,2\}$ which are associated with each vertex of the tree $\Gamma^k.$ The SOS model with nearest-neighbour and one-level next-nearest-neighbor interactions is defined by the following Hamiltonian:
\begin{equation}\label{eq1}
H(\sigma)=-J\sum_{\langle x,y \rangle }\mid\sigma(x)-\sigma(y)\mid-J_1\sum_{> \overline{x,y} < }\mid\sigma(x)-\sigma(y)\mid,
\end{equation}
where the sum in the first term ranges all nearest neighbors, second sum ranges all one-level next-nearest-neighbors, and $J,~J_1\in \mathbb{R}$ are the coupling constants (see Fig 1.).

\begin{center}
\includegraphics[width=12cm,height=5cm]{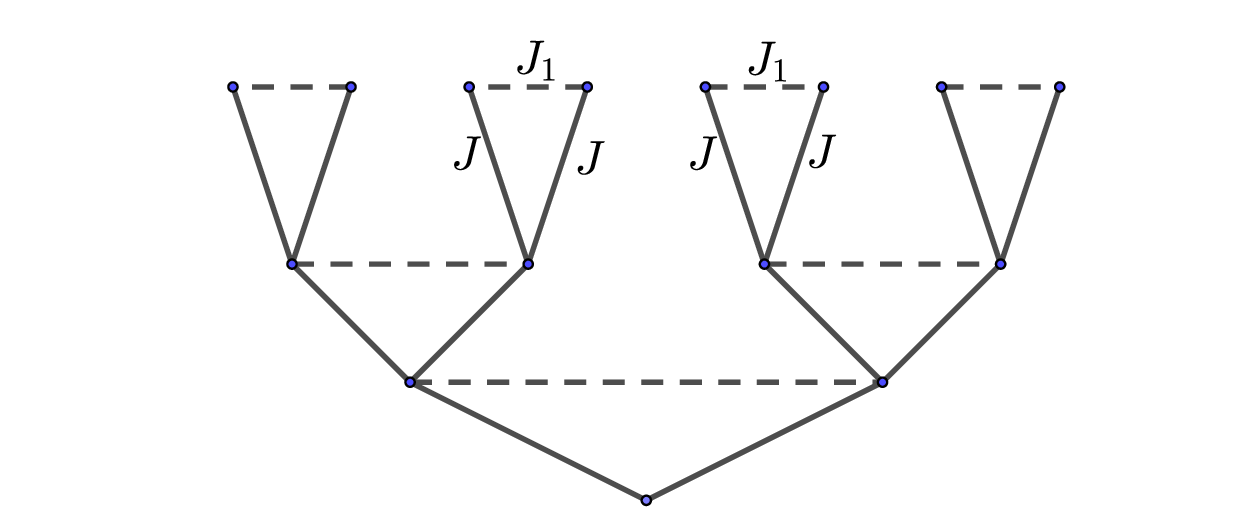}\label{Fig1}
\end{center}
\begin{center}{\footnotesize \noindent
 Figure 1.  The Cayley tree of order two with nearest-neighbor (---------) and one-level next-nearest-neighbour (- - - - -) interactions.}\
\end{center}

\section{Recursive Equations}\label{recur}
One can obtain the nonlinear functional equations (tree recursion) describing limiting Gibbs measures for lattice models on Cayley tree in many ways. One approach is based on properties of Markov random fields on Bethe lattices (see, e.g., \cite{rs}). The second approach is based on recursive equations for partition functions (see, e.g., \cite{Tur}). Naturally, both approaches lead to the same equation (see, e.g., \cite{Ro}). Since the second approach more suitable for models with competing interactions, we follow this approach.

Let $\Lambda$ be a finite subset of $V$. We will denote by $\sigma(\Lambda)$ the restriction of $\sigma$ to $\Lambda.$ Let $\overline{\sigma}(V\backslash\Lambda)$ be a fixed boundary configuration. The total energy of $\sigma(\Lambda)$ under condition $\overline{\sigma}(V\backslash\Lambda)$ is defined as
\begin{equation}\label{eq2}
\begin{split}
H(\sigma(\Lambda)\mid\overline{\sigma}(V\backslash\Lambda))=-J\sum_{\langle x,y \rangle:x,y\in \Lambda }\mid\sigma(x)-\sigma(y)\mid
\\-J_1\sum_{> x,y <:x,y\in\Lambda }\mid\sigma(x)-\sigma(y)\mid-J\sum_{\langle x,y \rangle:x\in \Lambda, y\notin\Lambda }\mid\sigma(x)-\sigma(y)\mid,
\end{split}
\end{equation}
Then partition function $Z_\Lambda(\overline{\sigma}(V\backslash\Lambda))$ in volume $\Lambda$ boundary condition $\overline{\sigma}(V\backslash\Lambda)$ is defined as
\begin{equation}\label{eq3}
Z_\Lambda(\overline{\sigma}(V\backslash\Lambda))=\sum_{\sigma(\Lambda)\in\Omega(\Lambda)}\exp(-\beta H_{\Lambda}(\sigma(\Lambda)\mid\overline{\sigma}(V\backslash\Lambda))),
\end{equation}
where $\Omega(\Lambda)$ is the set of all configurations in volume $\Lambda$ and $\beta=\frac{1}{T}$ is the inverse temperature. Then conditional Gibbs measure $\mu_{\Lambda}$ of a configuration $\sigma(\Lambda)$ is defined as
\[
\mu_{\Lambda}(\sigma(\Lambda)\mid\overline{\sigma}(V\backslash\Lambda))=\frac{\exp(-\beta\,H(\sigma(\Lambda)\mid\overline{\sigma}(V\backslash\Lambda)))}{Z_{\Lambda}(\overline{\sigma}(V\backslash\Lambda))}.
\]

We consider the configuration $\sigma(V_n)$, the partitions functions $Z_{V_n}$ and conditional
Gibbs measure $\mu_{\Lambda}(\sigma(\Lambda)\mid\overline{\sigma}(V\backslash\Lambda))$ in volume $V_n$ and for the sake of simplicity, we denote them by $\sigma_n$, $Z^{(n)}$ and $\mu_n,$ respectively. The partitions function $Z^{(n)}$ can be decomposed into following summands:
\begin{equation}\label{eq4}
Z^{(n)}= Z_0^{(n)}+Z_1^{(n)}+Z_2^{(n)},
\end{equation}
where
\begin{equation}\label{eq5}
Z_i^{(n)}=\sum_{\sigma_n\in\Omega(V_n):\sigma(x^0)=i}\exp(-\beta H_{V_n}(\sigma\mid\overline{\sigma}(V\backslash V_n))),~i=0,1,2.
\end{equation}
From now on, we restrict ourselves to the case $k=2.$

Denote $\theta=\exp(\beta J),~\theta_1=\exp(\beta J_1).$ Let $S(x^0)=\{x^1,~x^2\}.$ If $\sigma(x^0)=i,~\sigma(x^1)=j$ and $\sigma(x^2)=m,$ then from \eqref{eq2} and \eqref{eq3} we have following
$$
Z_i^{(n)}=\sum_{j,m=0}^2\exp(\beta J \mid i-j\mid+\beta J \mid i-m\mid +\beta J_1 \mid j-m\mid)Z_j^{n-1}Z_m^{(n-1)},
$$
so that
$$
Z_0^{(n)}=\Big[\big(Z_0^{(n-1)}\big)^2+2\theta\theta_1Z_0^{(n-1)}Z_1^{(n-1)}+2\theta^2\theta_1^2Z_0^{(n-1)}Z_2^{(n-1)}$$$$
+\theta^2\big(Z_1^{(n-1)}\big)^2+2\theta^3\theta_1Z_1^{(n-1)}Z_2^{(n-1)}+\theta^4\big(Z_2^{(n-1)}\big)^2\Big],
$$
$$
Z_1^{(n)}=\Big[\theta^2\big(Z_0^{n-1}\big)^2+2\theta\theta_1 Z_0^{(n-1)}Z_1^{(n-1)}+2\theta^2\theta_1^2Z_0^{(n-1)}Z_2^{(n-1)}$$
$$+\big(Z_1^{(n-1)}\big)^2+2\theta\theta_1Z_1^{(n-1)}Z_2^{(n-1)}+\theta^2\big(Z_2^{(n-1)}\big)\Big],
$$
$$
Z_2^{(n)}=\Big[\theta^4\big(Z_0^{(n-1)}\big)^2+2\theta^3\theta_1Z_0^{(n-1)}Z_1^{(n-1)}+2\theta^2\theta_1^2Z_0^{(n-1)}Z_2^{(n-1)}$$
$$+\theta^2\big(Z_1^{(n-1)}\big)^2+2\theta\theta_1Z_1^{(n-1)}Z_1^{(n-1)}+\big(Z_2^{n-1}\big)^2\Big].
$$
Introducing the notations $u_n(x^0)=\frac{Z_1^{(n)}(x^0)}{Z_0^{(n)}(x^0)},$ $v_n=\frac{Z_2^{(n)}(x^0)}{Z_0^{(n)}(x^0)},$ we obtain the following system of recurrent equations:
\begin{equation}\label{eq61}
\left\{%
\begin{array}{ll}
u_n=\frac{\theta^2+2\theta\theta_1u_{n-1}+2\theta^2\theta_1^2v_{n-1}+u_{n-1}^2+2\theta\theta_1u_{n-1}v_{n-1}+\theta^2v_{n-1}^2}{1+2\theta\theta_1u_{n-1}+2\theta^2\theta_1^2v_{n-1}+\theta^2u_{n-1}^2+2\theta^3\theta_1u_{n-1}v_{n-1}+\theta^4v_{n-1}^2}\\
[0.6cm]
v_n=\frac{\theta^4+2\theta^3\theta_1u_{n-1}+2\theta^2\theta_1^2v_{n-1}+\theta^2 u_{n-1}^2+2\theta\theta_1u_{n-1}v_{n-1}+v_{n-1}^2}{1+2\theta\theta_1u_{n-1}+2\theta^2\theta_1^2v_{n-1}+\theta^2u_{n-1}^2+2\theta^3\theta_1u_{n-1}v_{n-1}+\theta^4v_{n-1}^2}.\\
\end{array}%
\right.\end{equation}

Evidently,
\[
u_n(x^0)=\frac{\mu_n(\sigma_n(x^0)=1)}{\mu_n(\sigma_n(x^0)=0)},~v_n(x^0)=\frac{\mu_n(\sigma_n(x^0)=2)}{\mu_n(\sigma_n(x^0)=0)}.
\]
If we can find the limit of $u_n(x^0)$ as $n$ tends to infinity, we will find the ratio for the probability of a $1$ to the probability of a $0$ at the root for the limiting Gibbs measure. Similarly, if we can find the limit of $v_n(x^0)$ as $n$ tends to infinity, we will find the ratio for the probability of a $2$ to the probability of a $0$ at the root for the limiting Gibbs measure. Thus, the fixed points of the equation \eqref{eq6} describe the translation-invariant limiting Gibbs measure of the model \eqref{eq1}.

If $u=\lim u_n$ and $v=\lim v_n$ then
\begin{equation}\label{eq6}
\left\{%
\begin{array}{ll}
    u=\frac{\theta^2+2\theta\theta_1u+2\theta^2\theta_1^2v+u^2+2\theta\theta_1uv+\theta^2v^2}{1+2\theta\theta_1u+2\theta^2\theta_1^2v+\theta^2u^2+2\theta^3\theta_1uv+\theta^4v^2}, \\[0.6cm]
    v=\frac{\theta^4+2\theta^3\theta_1u+2\theta^2\theta_1^2v+\theta^2 u^2+2\theta\theta_1uv+v^2}{1+2\theta\theta_1u+2\theta^2\theta_1^2v+\theta^2u^2+2\theta^3\theta_1uv+\theta^4v^2}.\\
\end{array}%
\right.\end{equation}
\begin{remark} The system \eqref{eq6} coincides with the classical result for SOS model (see, e.g., \cite{rs,KRSOS}) when $\theta_1=1$ $(J_1=0)$, i.e.
\begin{equation}\label{equ}
\left\{%
\begin{array}{ll}
    u=\big(\frac{u+\theta v+\theta}{\theta^2v+\theta u+1}\big)^2, \\
    v=\big(\frac{\theta u+v+\theta^2}{\theta^2v+\theta u+1}\big)^2.\\
\end{array}%
\right.\end{equation}\end{remark}
It is important to note that if there is more than one positive solution for the system \eqref{eq6}, then there is more than one the translation-invariant limiting Gibbs measure corresponding to these solutions. We say that a phase transition occurs for the model \eqref{eq1}, if the system \eqref{eq6} has more than one positive solution.
\section{Ferromagnetic phases}\label{tigm}

In this section, we investigate the phase transitions of the model.  We consider the dynamic system \eqref{eq61} and study its limiting behaviour. Let $x=(u,v)\in\mathbb{R}_+^2$ and the dynamic system $F:\mathbb{R}_+^2\to\mathbb{R}_+^2$ is defined by
\begin{equation}\label{eq7}
\left\{%
\begin{array}{ll}
    u'=\frac{\theta^2+2\theta\theta_1u+2\theta^2\theta_1^2v+u^2+2\theta\theta_1uv+\theta^2v^2}{1+2\theta\theta_1u+2\theta^2\theta_1^2v+\theta^2u^2+2\theta^3\theta_1uv+\theta^4v^2}, \\[0.6cm]
    v'=\frac{\theta^4+2\theta^3\theta_1u+2\theta^2\theta_1^2v+\theta^2 u^2+2\theta\theta_1uv+v^2}{1+2\theta\theta_1u+2\theta^2\theta_1^2v+\theta^2u^2+2\theta^3\theta_1uv+\theta^4v^2}.\\
\end{array}%
\right.\end{equation}

Then the recurrent equations \eqref{eq61} be can rewritten as $x^{(n+1)}=F(x^{(n)}),~n\geq0.$ Recall that the point $x$ is a periodic point of period $p$ if $F^p(x) = x,$ where $F^p(x)$ stands for $p$-fold composition of $F$ into itself, i.e., $F^p(x)=\underbrace{F(F(\ldots F(x))\ldots)}_p$. A point $x\in\mathbb{R}_+^2$ is called a fixed point for $F:\mathbb{R}_+^2\to \mathbb{R}_+^2$ if
$F(x) = x.$ (see for more details \cite[Chapter 1]{Dev} or \cite[Section 1]{billiard}).

To investigate the problem of phase transition in the class of ferromagnetic phases, we have to describe the fixed points of the map $F(x)=x.$ Let us describe fixed points of this dynamical system, i.e., solutions of equation $F(x)=x.$ It is obvious that the following set is invariant with respect to the operator $F$:
 \begin{equation}\label{I}
I=\{x=(u,v)\in \mathbb{R}^2:v=1\}.
\end{equation}
On the set $I$, the system of equations \eqref{eq6} reduces to

 \begin{equation}\label{eq8}
u=f(u)
\end{equation}
where
\begin{equation}\label{f}
f(u)=f(u,\theta,\theta_1):=\frac{u^2+4\theta\,\theta_1\,u+2\theta^2(\theta_1^2+1)}{\theta^2\,u^2+2\theta\,\theta_1(\theta^2+1)u+\theta^4+2\,\theta^2\theta_1^2+1}.
\end{equation}
It is easy to see that the function $f(u)$ defined in \eqref{f} is continuous, bounded with $f(0)>0$ and $f(+\infty)<+\infty.$
From properties of the function $f$ it follows that the function $f$ has at least one fixed point, say, $u^*.$ We have

\begin{theorem}\label{thm3}
For the SOS model with one-level second nearest-neighbour interactions on the binary tree on the set $I$ if the condition $f'(u^*)>1$ is satisfied, then there exists three distinct translation-invariant limiting Gibbs measures, i.e., the phase transition occurs.
\end{theorem}
\begin{proof}
When $f'(u^*)>1,$ $u^*$ is unstable. So there exists a small neighborhood $(u^*-\varepsilon,u^*+\varepsilon)$ of $u^*$ such that for $u\in(u^*-\varepsilon,u^*)$ $f(u)<u,$ and for $(u^*,u^*+\varepsilon)$ $f(u)>u.$ Since $f(0)>0,$ there exists a solution between $0$ and $u^*.$ Similarly, since $f(+\infty)<+\infty$ there is another solution between $u^*$ and $+\infty.$ Thus, there exist three solutions. Since there exist a bijection between the solution of the Eq. \eqref{eq8} and the translation-invariant limiting Gibbs measures, it follows that there exist three translation-invariant limiting Gibbs measures, which implies the existence of a phase transition. This completes proof.
\end{proof}
\begin{remark}\label{rk2}
Note that the set of parameters which satisfy $f'(u^*)>1$ is not empty, e.g., see also Fig. 2.
\end{remark}
\begin{center}
\includegraphics[width=5cm,height=5cm]{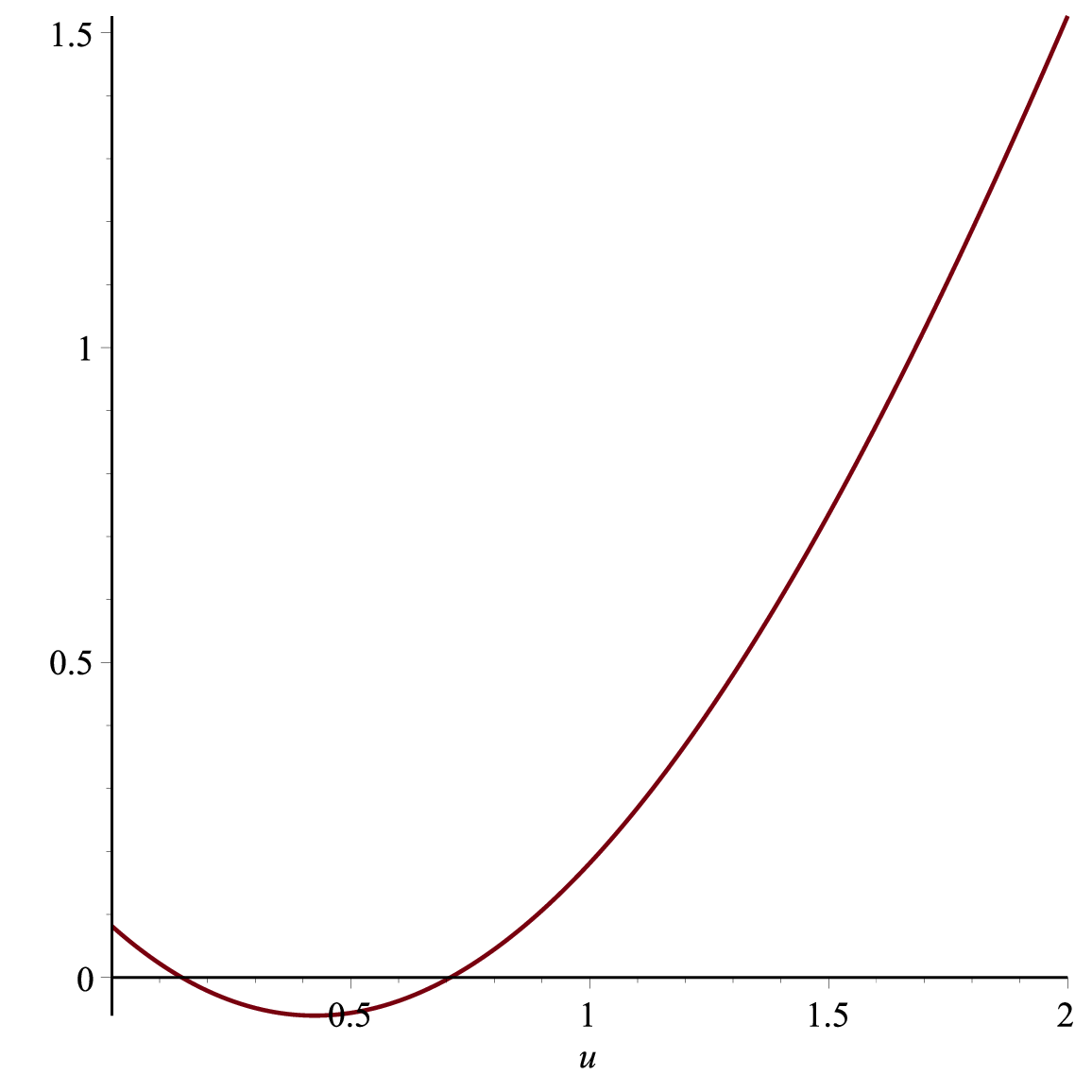}\label{Fig5}
\includegraphics[width=5cm,height=5cm]{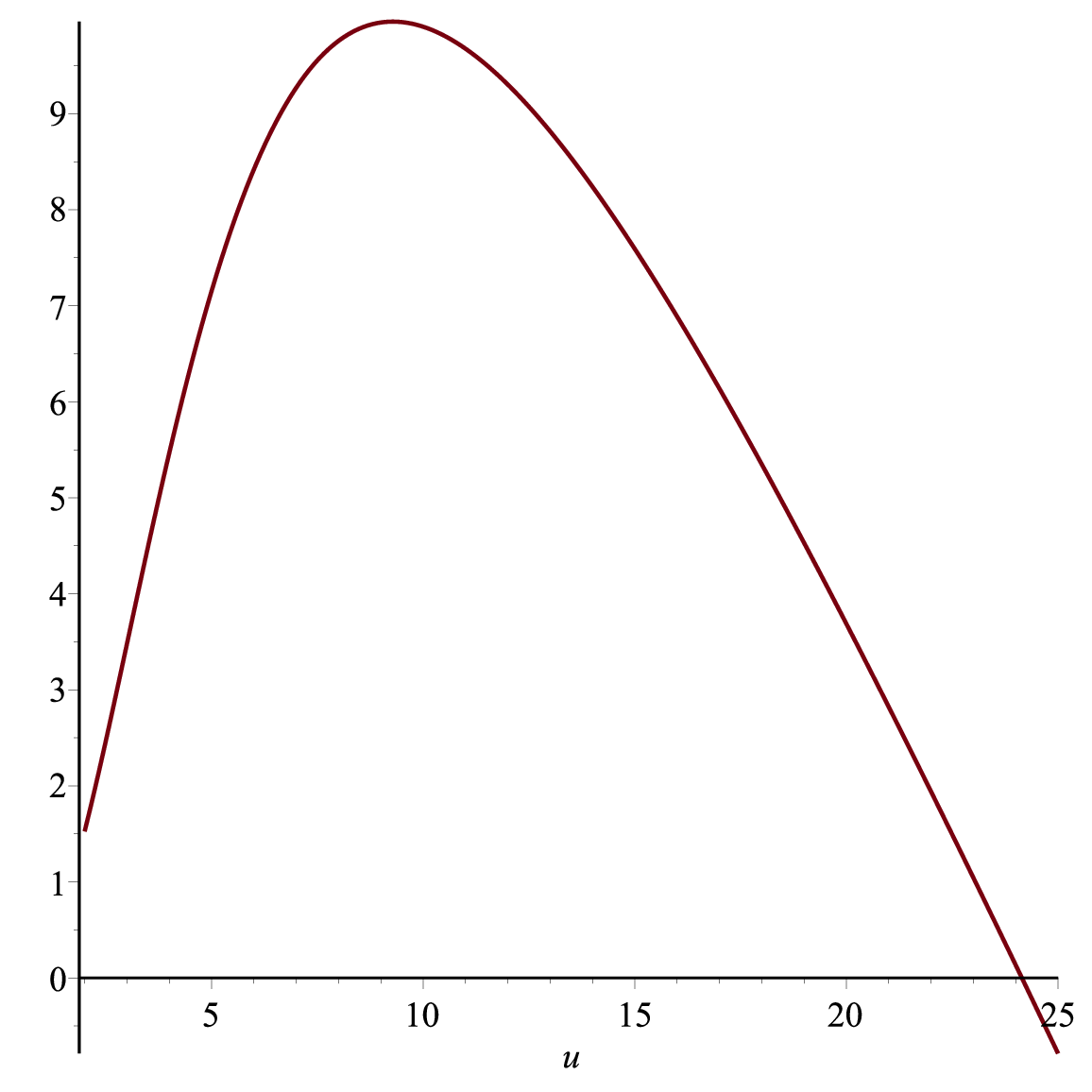}
\end{center}
\begin{center}{\footnotesize \noindent
Figure 2. The plot of $f(u)-u$ when $\theta=0.2,~\theta_1=0.5.$ In this case the function $f$ has three positive fixed points: $\approx{0.1461};~0.7085;~24.1453.$ The plot of the function is drawn for $u\in[0,2],~u\in[2,25]$ separately to show all 3 solutions.}\
\end{center}
\begin{remark}\label{rk11}
In Theorem \ref{thm3} we find the sufficient conditions for the Eq. \eqref{eq8} on possessing multiple solutions, i.e., there might be multiple solutions for the equation even if $f'(u^*)\leq 1.$
\end{remark}
\section{The absence of antiferromagnetic phases}\label{per}

 The fixed points of the transformation $x=F^2(x)$ that do not satisfy $x=F(x)$ define an antiferromagnetic phase, i.e., a phase with period 2. In the section, we study the antiferromagnetic phase of the model on the set $I$ \eqref{I}. Therefore, we investigate the zeroes of the equation given by

 \begin{equation}\label{funct}
\frac{f(f(u))-u}{f(u)-u}=0,
\end{equation}
where $f$ is defined in \eqref{f}.
Simplifying above equation, we obtain
\begin{equation}\label{eq14}
A\,u^2+B\,u+C=0
\end{equation}
where
\[
A:=A(\theta;\theta_1)=\theta^6+2\theta^4\theta_1^2+2\theta^3\theta_1+\theta^2+2\theta\theta_1+1,
\]
\[
B:=B(\theta;\theta_1)=2\theta^7\theta_1+4\theta^5\theta_1^3+2\theta^5\theta_1+6\theta^4\theta_1^2+4\theta^3\theta_1^3-\theta^4+2\theta^3\theta_1+10\theta^2\theta_1^2+6\theta\theta_1+1,
\]
\[
C:=C(\theta;\theta_1)=\theta^8+4\theta^6\theta_2+4\theta^4\theta_1^4+4\theta^5\theta_1+8\theta^3\theta_1^3+2\theta^4+6\theta^2\theta_1^2+2\theta^2+4\theta\theta_1+1.
\]
Note that $A>0,C>0$ for any $\theta>0,\theta_1>0.$ According to Descartes' Rule of Signs (see, e.g., \cite{PP}, Corollary 1) if $B\geq0$ then the equation \eqref{eq14} does not have any positive solution (see Fig. 3). Thus, we have the following assertion:

\begin{theorem}\label{thm45}
If \[
(\theta, \theta_1)\in\{(\theta, \theta_1)\in\mathbb{R}_+^2: B\geq0\}\]
then for the SOS model with one-level next-nearest-neighbour interactions on the binary tree there is no antiferromagnetic phases (two-periodic Gibbs measures) on the set $I$ \eqref{I}.
\end{theorem}

\begin{center}
\includegraphics[width=10cm,height=7cm]{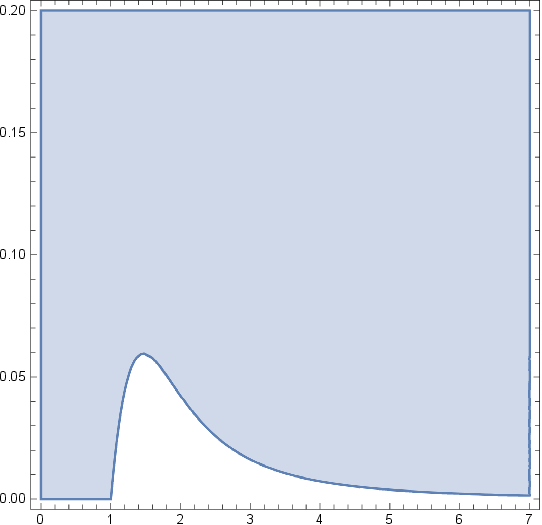}\label{Fig9}
\end{center}
\begin{center}{\footnotesize \noindent
Figure 3. The plot of $B(\theta,\theta_1)$ for $\theta\in (0,7)$ and $\theta_1\in(0,0.2)$ The shaded area corresponds to $B(\theta,\theta_1)\geq 0$.}\
\end{center}

Due to Theorem \ref{thm45}, we should consider the case $B<0.$ If $B<0$ then the equation \eqref{eq14} might have 2 positive solutions. We find the discriminant of the Eq. \eqref{eq14}:
\[
D:=D(\theta;\theta_1)=B^2-4\,A\,C.
\]
It is easy to see that if $B<0$ and $D\geq0$ then the Eq \eqref{eq14} has at least positive solution. However, a computer analysis shows that the set
\[
\mathcal{S}=\{(\theta,\theta_1)\in \mathbb{R}_{+}^2:~D\geq0,~B<0\}
\]
is empty.
Summarising, we make

\textbf{Conjecture 1.} \textit{The SOS model with one-level next-nearest-neighbour interactions on the binary tree does not have any antiferromagnetic phase (two-periodic Gibbs measures) on the set $I$ \eqref{I}.}

\begin{remark}\label{rmk4}
Note that for the model \eqref{eq1} there might be antiferromagnetic phases outside of the set $I$ \eqref{I}.
\end{remark}

\section{Declarations}
\subsection{Author contributions} Conceptualization: M.~M.~Rahmatullaev; Methodology: M.~M.~Rahmatullaev; Formal analysis and investigation: O.~Sh.~Karshiboev; Writing - original draft preparation: O.~Sh.~Karshiboev; Writing - review and editing: O.~Sh.~Karshiboev; Resources: M.~M.~Rahmatullaev, Supervision: M.~M.~Rahmatullaev.
\subsection{Funding} No funds, grants, or other support was received
\subsection{Competing interests} The authors have no relevant financial or non-financial interests to disclose
\subsection{Data Availability Statement} Not applicable


\begin{thebibliography}{99}

\bibitem{Dev} Devaney, R. L.: An introduction to chaotic dynamical system, Westview Press, (2003)

\bibitem{Tur} Ganikhodjaev, N. N., Akin, H., Temir, S.: Potts model with two competing binary interactions, Turkish Journal of Mathematics, Vol. 31: No. 3, (2007) https://journals.tubitak.gov.tr/math/vol31/iss3/1

\bibitem{Gan06} Ganikhodjaev, N., Mukhamedov, F., Mendes, J. F. F.: On the three state Potts model with competing interactions on the Bethe lattice, J. Stat. Mech. (2006) https://doi.org/10.1088/1742-5468/2006/08/P08012

\bibitem{Pah2} Ganikhodjaev, N. N., Mukhamedov, F., Pah, C. H.: Phase diagram of the three states Potts model with next nearest neighbour interactions on the Bethe lattice, Physics Letters A, 1 (373), 33-38, (2008), https://doi.org/10.1016/j.physleta.2008.10.060

\bibitem{ganipah} Ganikhodjaev, N. N., Pah, C. H., Wahiddin, M. R. B.: Exact solution of an Ising model with competing interactions on a Cayley tree, Journal of Physics A: Mathematical and General, V.36, No. 15, pp. 4283-4289 (2003) https://doi.org/10.1088/0305-4470/36/15/305

\bibitem{GTA} Ganikhodjaev, N. N., Temir, S., Akin, H.: Modulated phase of a Potts model with competing binary interactions on a Cayley tree, J. Stat. Phys. 137, 701-715, (2009) https://doi.org/10.1007/s10955-009-9869-z

\bibitem{G} Georgii, H.-O.: Gibbs Measures and Phase Transitions, W. de Gruyter, Berlin, (1988). https://doi.org/10.1515/9783110250329

\bibitem{tmph} Karshiboev, O. Sh.: Periodic Gibbs measures for the three-state SOS model on a Cayley tree with a translation-invariant external field. Theor Math Phys 212, 1276-1283 (2022). https://doi.org/10.1134/S0040577922090094

\bibitem{Katsura} Katsura, S. Takizawa, M.: Bethe lattice and the Bethe approximation, Progress of Theoretical Physics, Vol.51, Issue 1, 82-98, (1974) https://doi.org/10.1143/PTP.51.82

\bibitem{KRSOS} Kuelske, C., Rozikov, U. A.: Extremality of translation-invariant phases for a three-state SOS-model on the binary tree. J.Stat.Phys. (2015) 160:659-680, https://doi.org/10.1007/s10955-015-1279-9.

\bibitem{k} Leung, K. T., Mok, I. A. C., Seun, S. N.: Polynomials and equations, Hong Kong University Press, (1992).

\bibitem{PP} Prasolov, V. V.: Polynomials. Springer Science \& Business Media, (2004). https://doi.org/10.1007/978-3-642-03980-5

\bibitem{B} Preston, C. J.: Gibbs states on countable sets, Cambridge University Press, London (1974), https://doi.org/10.1017/CBO9780511897122

\bibitem{Bun} Rahmatullaev, M. M., Abraev, B. U.: Non-translation invariant Gibbs measures of an SOS model on a Cayley tree, Reports on mathematical physics. Vol 86, No.3, (2020)

\bibitem{Sib} Rahmatullaev M. M., Abraev B.U.: On ground states for the SOS model with competing interactions, Journal of Siberian Federal University, Vol 15. No.2, 1-14, (2022)

\bibitem{Pos} Rahmatullaev, M. M., Karshiboev, O. Sh.: Gibbs measures for the three-state SOS model with external field on a Cayley tree, Positivity \textbf{26},(74), 1--15, (2022), https://doi.org/10.1007/s11117-022-00940-y

\bibitem{PT} Rahmatullaev, M. M., Karshiboev, O. Sh.: Phase transition for the SOS model under inhomogeneous external field on a Cayley tree, Phase Transitions, 95:12, 901-907, (2022) https://doi.org/10.1080/01411594.2022.2138756

\bibitem{billiard} Rozikov, U. A.: An introduction to mathematical billiards, World Sci., Hackensack, NJ., (2019), https://doi.org/10.1142/11162

\bibitem{Ro} Rozikov, U. A.: Gibbs Measures on Cayley Trees, World Scientific, Singapore, (2013), https://doi.org/10.1142/8841

\bibitem{shok} Rozikov, U. A., Shoyusupov, Sh. A.: Gibbs measures for the SOS model with four states on a Cayley tree, Theor. Math. Phys., 149(1): 1312-1323, (2006), https://doi.org/10.1007/s11232-006-0120-7

\bibitem{rs} Rozikov, U.A., Suhov, Y. M.: Gibbs measures for SOS model on a Cayley tree. Inf.Dim.An.,Quant.Prob. and Related Topics. Vol.9, No.3, 471-488, (2006), https://doi.org/10.1142/S0219025706002494

\bibitem{Trag} Tragtenberg, M. H. R., Yokoi, C. S. O.: Field behaviour of an Ising model with competing interactions on the
Bethe lattice, Phys. Rev. E, (1995), https://doi.org/10.1103/PhysRevE.52.2187

\bibitem{Van} Vannimenus, J.: Modulated phase of an Ising system with competing interactions on a Cayley tree, Z. Phys. B 43, 141-148, (1981), https://doi.org/10.1007/BF01293605

\bibitem{Yokoi} Yokoi, C. S. O., Oliveira, M. J., Salinas, S. R.: Strange attractor in the Ising model with competing interactions on the Cayley tree, Phys. Rev. Lett., (1985), https://doi.org/10.1103/PhysRevLett.54.163















\end{thebibliography}
\end{document}